\documentclass[conference]{IEEEtran}

\usepackage{graphicx}
\usepackage{amsmath}
\newtheorem{theorem}{Theorem}
\newtheorem{lemma}{Lemma}
\newtheorem{definition}{Definition}
\newtheorem{property}{Property}

\begin{document}

\title{Physical Network Coding in Two--Way Wireless Relay Channels}
\author{\authorblockN{Petar Popovski and Hiroyuki Yomo}\\
\authorblockA{Department of Electronic Systems, Aalborg University\\
Niels Jernes Vej 12, DK-9220 Aalborg, Denmark\\
Email: \{petarp, yomo\}@kom.aau.dk}}

\maketitle

\begin{abstract}
It has recently been recognized that the wireless networks
represent a fertile ground for devising communication modes based
on network coding. A particularly suitable application of the
network coding arises for the two--way relay channels, where two
nodes communicate with each other assisted by using a third, relay
node. Such a scenario enables application of \emph{physical
network coding}, where the network coding is either done (a)
jointly with the channel coding or (b) through physical combining
of the communication flows over the multiple access channel. In
this paper we first group the existing schemes for physical
network coding into two generic schemes, termed 3--step and
2--step scheme, respectively. We investigate the conditions for
maximization of the two--way rate for each individual scheme: (1)
the Decode--and--Forward (DF) 3--step schemes (2) three different
schemes with two steps: Amplify--and--Forward (AF), JDF and
Denoise--and--Forward (DNF). While the DNF scheme has a potential
to offer the best two--way rate, the most interesting result of
the paper is that, for some SNR configurations of the
source---relay links, JDF yields identical maximal two--way rate
as the upper bound on the rate for DNF.
\end{abstract}

\section{Introduction}

It has been recently noted~\cite{ref:Fragouli} that broadcast and
unreliable nature of the wireless medium sets a fertile ground for
developing network--coding~\cite{ref:AhlswedeNC} solutions. The
network coding can offer performance improvement in the wireless
networks for two--way (or multi--way) communication
flows~\cite{ref:Chou_conf} \cite{ref:Ericsson_DF_BAT}
\cite{ref:BAT-VTC} \cite{ref:BAT-ICC} \cite{ref:Hausl}
\cite{ref:Katti} \cite{ref:Rankov}. In general, there are two
generic schemes for two--way wireless relay (Fig.~\ref{fig:scen}):
(a) 3--step scheme (b) 2--step scheme.
\begin{figure}
    \centering
    \includegraphics[width=8 cm]{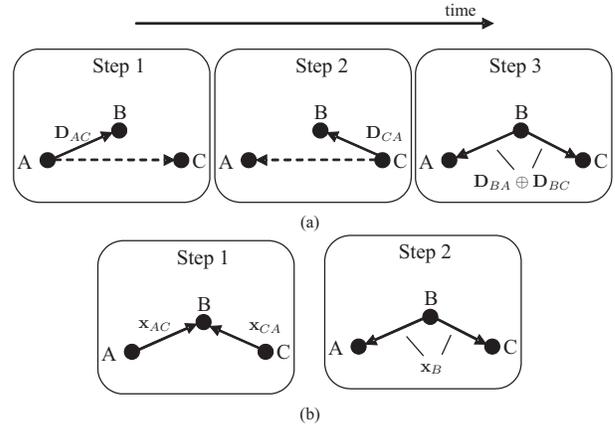}\\
    \caption{Generic schemes for physical network coding over the two--way relay channel. (a)
    Three--step scheme (b) 2--step scheme.
    }\label{fig:scen}
    \vspace{-12pt}
\end{figure}
The node $A$ has packets for the node $C$ and vice versa. In Step
1 of the 3--step scheme, $A$ transmits the packet
$\textbf{D}_{AC}$, in Step 2 $C$ transmits the packet
$\textbf{D}_{CA}$. Here $B$ decodes both packets, such that the
3--step schemes are \emph{Decode--and--Forward (DF)} schemes. In
the simpler DF schemes~\cite{ref:Chou_conf}
\cite{ref:Ericsson_DF_BAT} \cite{ref:BAT-VTC}, the direct link
between $A$ and $C$ is \emph{ignored} by the receivers in Steps 1
and 2, such that in Step 3 $B$ broadcasts the packet
$\textbf{D}_{BC} \oplus \textbf{D}_{BA} = \textbf{D}_{AC} \oplus
\textbf{D}_{CA}$, where $\oplus$ is \verb|XOR| operation, after
which the node $A$ ($C$) is able to decode the packet
$\textbf{D}_{CA} (\textbf{D}_{AC})$. While it is hard to
characterize such a simple DF scheme as ``physical'' network
coding, such an attribute can be attached to the 3--step DF
scheme~\cite{ref:Hausl}, where the direct link $A-C$ is not
ignored in the Steps 1 and 2 and a joint network--channel coding
is needed. In that case, the packet $\textbf{D}_{BA}
(\textbf{D}_{BC})$ is a many--to--one function of the packet
$\textbf{D}_{CA} (\textbf{D}_{AC})$, since $A$ ($C$) already has
some information from the Step 2 (1). In the 2--step schemes the
communication flows are combined through a simultaneous
transmission over a multiple access channel. In Step 1 $B$
receives a noisy signal that consists of interference between the
signals of $A$ and $C$. Due to the half--duplex operation, the
direct link is naturally ignored in the 2--step schemes. The
signal $\textbf{x}_B$ that is broadcasted in Step 2 depends on the
applied 2--step scheme. In \emph{Amplify--and--Forward
(AF)}~\cite{ref:BAT-VTC}, $\textbf{x}_B$ is simply an amplified
version of the signal received by $B$ in step 1. After receiving
$\textbf{x}_B$, the node $A$ ($C$) subtracts its own signal and
decodes the signal sent by $C$ ($A$) in Step 1. The 2--step scheme
termed \emph{Denoise--and--Forward (DNF)} has been introduced
in~\cite{ref:BAT-ICC}. A related scheme appeared
in~\cite{ref:Xiao_nested}. In DNF, the node $B$ again does not
decode the packets sent by $A$ and $C$ in Step 1, but it maps the
received signal to a codeword from a discrete set. Hence, the
signal $\textbf{x}_B$ carries now the information about the set of
codeword pairs $\{ (\textbf{x}_{AC},\textbf{x}_{CA}) \}$ which are
considered by the node $B$ as likely to have been sent in the Step
1. In general, this set can consist of several codeword pairs,
such that $B$ has an ambiguity which information has been sent.
Nevertheless, since $A$ ($C$) knows $\textbf{x}_{AC}
(\textbf{x}_{CA})$, after receiving $\textbf{x}_B$, it will
extract exactly one codeword as a likely one to have been sent by
$C$ ($A$) in Step 1. The final considered 2--step scheme is
\emph{Joint Decode--and Forward (JDF)}, recently considered
in~\cite{ref:Rankov}. In JDF, the transmission rates in Step 1 of
Fig.~\ref{fig:scen}(b) are selected such that $B$ can jointly
decode both $\textbf{x}_{AC}$ and $\textbf{x}_{CA}$, and then use
\verb|XOR| to obtain the signal for broadcast in Step 2.

In this paper we investigate the strategies that can maximize the
overall two--way rate for several 2-- and 3--step schemes for
physical network coding.  We show that the key to maximizing the
two--way rate in the system for the 3--step schemes is the
relation between the durations of Step 1 and Step 2. On the other
hand, we show that the key factor for maximizing the two--way rate
in the 2--step schemes is the choice of the rates  at which $A$
and $C$ transmit in Step 1. Note that we are not providing the
absolute capacities of the two--way relay channel, since we are
putting some operational restrictions to the applied schemes.
Nevertheless, the results give an excellent overview of what can
be achieved by each scheme for physical network coding.

\section{Notations and Definitions}
We assume that there are only two communication flows,
$A\rightarrow C$ and $C\rightarrow A$, respectively. The relay $B$
is neither a source nor a sink of any data in the system. All the
nodes are half--duplex, such that a node can either transmit or
receive at a given time. We use $x_U[m]$ to denote the $m-$th
complex baseband transmitted symbol from node $U \in \{A, B, C\}$.
A complex--valued vector is denoted by $\mathbf{x}$. A packet of
bits is denoted by $\textbf{D}$, and the number of bits in the
packet is $|\textbf{D}|$. If only one node $U \in \{A, B, C \}$ is
transmitting, then the $m-$th received symbol at the node $V \in
\{A, B, C \} \setminus U$ is given by:
\begin{equation}\label{eq:ReceivedVtoU}
    y_V[m]=h_{UV}x_U[m]+z_V[m]
\end{equation}
where $h_{UV}$ is the complex channel coefficient between $U$ and
$V$. $z_V[m]$ is the complex additive white Gaussian noise ${\cal
CN}(0,N_0)$. The transmitted symbols have $E\{x_U[m]\}=0$ and a
normalized power $E\{|x_U[m]|^2\}=1$. Each node uses the same
transmission power, which makes the links symmetric: {\small
\begin{equation}\label{eq:ChannelGains}
    h_{AC}=h_{CA}=h_0; \quad h_{AB}=h_{BA}=h_1; \quad
    h_{CB}=h_{BC}=h_2
\end{equation}
}We consider time--invariant channels and $h_0, h_1, h_2$ are
perfectly known by all nodes. This assumption allows us to find
the two--way rates at which a reliable communication is possible.
The bandwidth is normalized, such that we consider the following
signal--to--noise ratios (SNRs):
\begin{equation}\label{eq:SNRs}
    \gamma_i=\frac{|h_i|^2}{N_0} \qquad i=0,1,2
\end{equation}
The bandwidth is normalized to 1 Hz, such that a link with SNR of
$\gamma$ can reliably transfer up to:
\begin{equation}\label{eq:Definition of C gamma}
    C(\gamma)=\log_2(1+\gamma)\textrm{ [bit/s]}
\end{equation}
The time is measured in number of symbols, such that when a packet
of $N$ symbols is sent at the data rate $r$, the packet contains
$Nr$ bits. The packet lengths are sufficiently large, such that we
can use codebooks that offer zero errors if the rate is chosen to
be below the channel capacity.

Without loss of generality, we assume that
\begin{equation}\label{eq:RelationSNR1SNR2}
    \gamma_2 \geq \gamma_1
\end{equation}
The source--to--relay links are assumed better than the direct
link~\cite{ref:CoverRelay}:
\begin{equation}\label{eq:Relation to SNR0}
    \gamma_1 > \gamma_0 \qquad \gamma_2 > \gamma_0
\end{equation}
If $A$ and $C$ transmit simultaneously, then $B$ receives:
\begin{equation}\label{eq:MAchannel_at_B}
    y_B[m]=h_1x_A[m]+h_2x_B[m]+z_B[m]
\end{equation}

In this paper we will be interested in the \textbf{two--way rate}:
\begin{definition}
Let, during a time of $N$ symbols, $A$ receive reliably
$|\textbf{D}_{CA}|$ bits from $C$ and $C$ receive reliably
$|\textbf{D}_{AC}|$ bits from $A$. Then the two--way rate is given
by:
\begin{equation}\label{eq:def two-way rate}
    R_{A \leftrightarrow
    C}=\frac{|\textbf{D}_{AC}|+|\textbf{D}_{CA}|}{N}\textrm{
    [bits/s]}
\end{equation}
\end{definition}

We seek to maximize the two--way rate under the following two
operational restrictions. \emph{First}, in each round $A$ and $C$
transmit only fresh data, which is independent of any information
exchange that took part in the previous rounds. \emph{Second}, $B$
is applying potentially suboptimal broadcast strategy, as we have
not explicitly considered the broadcast strategies that achieve
the full capacity region of the Gaussian broadcast
channel~\cite{ref:ThomasCover}. Hence, the obtained two--way rates
are lower bounds on the achievable rates in the two--way relay
systems.

\section{3--Step Scheme}

A single round in a 3--step scheme is (Fig.~\ref{fig:scen}(a)):
\textbf{Step 1:} Node $A$ transmits, nodes $B$ and $C$ receive.
\textbf{Step 2:} Node $C$ transmits, nodes $A$ and $B$ receive.
\textbf{Step 3:} Node $B$ transmits, nodes $A$ and $C$ receive. In
this scheme, $B$ should decode the data transmitted by node $A$
(node $C$) in Step 1 (Step 2). The data transmitted by $C$ in Step
2 is independent of the data received from $A$ in Step 1. The data
transmitted by the node $B$ in Step 3 is a function of the data
that was transmitted by $A$ and $C$ in Step 1 and 2, respectively,
from the \emph{same round}.

We first determine the size of the data broadcasted by $B$. If $A$
is transmitting $K$ symbols at a data rate $C(\gamma_1)$, then $B$
receives reliably the packet $\textbf{D}_{AC}$ of $KC(\gamma_1)$
bits. At the same time, the total amount of information received
at the node $C$ is $KC(\gamma_0)$ bits, where
$C(\gamma_0)<C(\gamma_1)$, due to~(\ref{eq:Relation to SNR0}).
Hence, in the next step the relay needs to transmit at least:
\begin{equation}\label{eq:Bits to be relayed}
    |\textbf{D}_{BC}|=K[C(\gamma_1)-C(\gamma_0)]
\end{equation}
bits to $C$ in order to completely remove the uncertainty at $C$
about the message transmitted by $A$. It is crucial to note that
the node $A$ \emph{knows the content of the packet}
$\textbf{D}_{BC}$. The argument to show this is that, after $B$
receives $\textbf{D}_{AC}$, both $A$ and $B$ have the same
information $\textbf{D}_{AC}$ and no information what has been
received at $C$. Even then, the \emph{random binning}
technique~\cite{ref:ThomasCover} can be used to create the packet
$\textbf{D}_{BC}$, such that $\textbf{D}_{BC}$ is uniquely and in
advance determined for each $\textbf{D}_{AC}$.

Let the node $A$ in Step 1 transmit a packet $\textbf{D}_{AC}$ of
$N(1-\theta)$ symbols at a rate $C(\gamma_1)$, where $0<\theta<1$.
Upon successfully decoding $\textbf{D}_{AC}$, the relay node $B$
prepares $\textbf{D}_{BC}$ that needs to be forwarded to $C$, with
a packet size of:
\begin{equation}\label{eq:PacketSize DBC}
    |\textbf{D}_{BC}|=N(1-\theta)\log_2 \left[
    C(\gamma_1)-C(\gamma_0) \right])\textrm{ [bits]}
\end{equation}
During the next $N\theta$ symbols, in Step 2, the node $C$
transmits $\textbf{D}_{CA}$ at a rate $C(\gamma_2)$, out of which
$B$ creates $\textbf{D}_{BA}$ with:
\begin{equation}\label{eq:PacketSize DBA}
    |\textbf{D}_{BA}|=N\theta \log_2 \left[
    C(\gamma_2)-C(\gamma_0)\right]\textrm{ [bits]}
\end{equation}
It follows from above that  $A$ knows $\textbf{D}_{BC}$ and $C$
knows $\textbf{D}_{BA}$. In addition, the node $A$ does not know
$\textbf{D}_{BA}$, but it knows a priori the size of the packet
$|\textbf{D}_{BA}|$. The same is valid for $C$ and the packet size
$|\textbf{D}_{BC}|$. This is reasonable for the assumed
time--invariant systems with fixed $h_0,h_1,h_2$.

\noindent\begin{theorem} \label{theorem:DF} The maximal two--way
rate for DF is
\begin{equation}\label{eq:MaxRateDF}
    R^{*}_{DF}=C(\gamma_1)\frac{1+\delta[C(\gamma_2)-C(\gamma_1)]}{1+\delta[C(\gamma_2)-C(\gamma_0)]}
\end{equation}
where
$\delta=\frac{[C(\gamma_1)-C(\gamma_0)]}{C(\gamma_1)[C(\gamma_1)+C(\gamma_2)-2C(\gamma_0)]}$.
\end{theorem}

\begin{proof}
In Step 3, the node $B$ first compares the packet sizes
$|\textbf{D}_{BC}|$ and $|\textbf{D}_{BA}|$. Two cases can occur:
\subsubsection{Case 1: $|\textbf{D}_{BC}|\geq
|\textbf{D}_{BA}|$}
Using~(\ref{eq:PacketSize DBC})
and~(\ref{eq:PacketSize DBA}), we can translate this condition
into inequality for $\theta$:
\begin{equation}\label{eq:Inequality theta 1}
    0 < \theta \leq \frac{C(\gamma_1)-C(\gamma_0)}{C(\gamma_1)+C(\gamma_2)-2C(\gamma_0)}
\end{equation}
The relay $B$ partitions the packet $\textbf{D}_{BC}$ into
$\textbf{D}^{(1)}_{BC}$ and $\textbf{D}^{(2)}_{BC}$:
\begin{equation}\label{eq:PacketSizes Division DBC}
    |\textbf{D}^{(1)}_{BC}|=|\textbf{D}_{BA}| \qquad |\textbf{D}^{(2)}_{BC}|=|\textbf{D}_{BC}|-|\textbf{D}_{BA}|
\end{equation}
$\textbf{D}^{(1)}_{BC}$ consists of the first $|\textbf{D}_{BA}|$
bits from $\textbf{D}_{BC}$ and $\textbf{D}^{(2)}_{BC}$ consists
of the rest of the bits from $\textbf{D}_{BC}$. Now $B$ creates:
\begin{equation}\label{eq:the XORpacket no paddding}
    \textbf{D}_{B}=\textbf{D}^{(1)}_{BC} \oplus \textbf{D}_{BA}
\end{equation}
where $\oplus$ is bitwise \verb|XOR|. Due to the
condition~(\ref{eq:RelationSNR1SNR2}) and the fact that both $A$
and $C$ need to receive it, the packet $\textbf{D}_{B}$ is
transmitted at the lower rate $C(\gamma_1)$. After receiving
$\textbf{D}_{B}$, the node $A$ extracts the packet
$\textbf{D}_{BA}$ as $\textbf{D}_{BA}=\textbf{D}_{B} \oplus
\textbf{D}^{(1)}_{BC}$. This packet is then used together with the
information that $A$ has received from node $C$ in Step 2 to
decode the packet $\textbf{D}_{CA}$. On the other hand, after
receiving $\textbf{D}_{B}$, the node $C$ extracts
$\textbf{D}^{(1)}_{BC}=\textbf{D}_{B} \oplus \textbf{D}_{BA}$. Now
$B$ transmits the packet $\textbf{D}^{(2)}_{BC}$ to the node $C$
at a higher rate of $C(\gamma_2)$, as $A$ does not need to receive
this information. With $\textbf{D}^{(2)}_{BC}$ and
$\textbf{D}^{(1)}_{BC}$, the node $C$ creates $\textbf{D}_{BC}$,
which is further on used jointly with the information that $C$ has
received in Step 1 to decode the packet $\textbf{D}_{AC}$. The
total duration of the three steps is $N_{1,
DF}(\theta)=N(1-\theta)+N\theta+\frac{|\textbf{D}_{BA}|}{C(\gamma_1)}+\frac{|\textbf{D}_{BC}|-|\textbf{D}_{BA}|}{C(\gamma_2)}$,
resulting in a two--way rate of:
\begin{equation}\label{eq:DataRate_DF_1}
    R_{1,DF}(\theta)=\frac{|\textbf{D}_{AC}|+|\textbf{D}_{CA}|}{N_{1,DF}}\textrm{ [bits/s]}
\end{equation}
where $|\textbf{D}_{BC}|$ and $|\textbf{D}_{BA}|$ are functions of
$\theta$ and are given by~(\ref{eq:PacketSize DBC})
and~(\ref{eq:PacketSize DBA}), respectively. It can be proved that
$R_{1,DF}(\theta)$ is monotonically increasing function of
$\theta$, such that $R_{1,DF}(\theta)$ achieves its maximal value
for the upper limiting value of $\theta$, given
in~(\ref{eq:Inequality theta 1}). By applying
$\theta=\frac{C(\gamma_1)-C(\gamma_0)}{C(\gamma_1)+C(\gamma_2)-2C(\gamma_0)}$
into the terms of~(\ref{eq:DataRate_DF_1}), we obtain the two--way
rate given by~(\ref{eq:MaxRateDF}).

\subsubsection{Case 2: $|\textbf{D}_{BC}|<
|\textbf{D}_{BA}|$} This is the region:
\begin{equation}\label{eq:Inequality theta 2}
   \frac{C(\gamma_1)-C(\gamma_0)}{C(\gamma_1)+C(\gamma_2)-2C(\gamma_0)} < \theta
   \leq 1
\end{equation}
The packet $\textbf{D}_{BC}$ is padded with zeros to obtain the
packet $\textbf{D}^{p}_{BC}$ such that
$|\textbf{D}^{p}_{BC}|=|\textbf{D}_{BA}|$. Since $A$ and $C$ know
the size of $|\textbf{D}_{BC}|$, they also know how many zeros are
used for padding. The node $B$ creates the packet
$\textbf{D}_{B}=\textbf{D}^{p}_{BC}\oplus \textbf{D}_{BA}$. In
Step 3 only the packet $\textbf{D}_{B}$ is broadcasted at a
transmission rate $C(\gamma_1)$. The node $A$ extracts
$\textbf{D}_{BA}$ as $\textbf{D}_{BA}=\textbf{D}^{p}_{BC}\oplus
\textbf{D}_{B}$ and uses the information received in Step 2 to
decode $\textbf{D}_{CA}$. Similarly, $C$ obtains
$\textbf{D}^{p}_{BC}$ from $\textbf{D}_{B}$, removes the padding
zeros and obtains $\textbf{D}_{BC}$, which is then used jointly
with the information from Step 1 to decode the packet
$\textbf{D}_{AC}$.

The total number of symbols is $N_{2,
DF}(\theta)=N(1-\theta)+N\theta+\frac{|\textbf{D}_{AB}|}{C(\gamma_1)}$
and the two--way rate $R_{2,DF}(\theta)$ is again calculated by
using the expression~(\ref{eq:DataRate_DF_1}), by putting $N_{2,
DF}$ instead of $N_{1, DF}$. It can be proved that
$R_{2,DF}(\theta)$ decreases monotonically with $\theta$ and it
reaches maximal value for the minimal $\theta$ in the
region~(\ref{eq:Inequality theta 2}). Hence, the maximal two--way
rate is again given by~(\ref{eq:MaxRateDF}).
\end{proof}

It can be seen that due to the condition~(\ref{eq:Relation to
SNR0}), the two--way rate is $R^{*}_{DF}<C(\gamma_1)$. When
$\gamma_1=\gamma_2$, the obtained capacity expression is identical
to what can be obtained from~\cite{ref:Hausl}. When $A$ and $C$
neglect the transmission over the direct link ($\gamma_0=0$), the
two--way rate achieved by DF is:
\begin{equation}\label{eq:MaxRateDF0}
    R^{0}_{DF}=\frac{2C(\gamma_1)C(\gamma_2)}{C(\gamma_1)+2C(\gamma_2)}
\end{equation}

\smallskip

\section{2--Step Schemes}

In this section we deal with three schemes: \emph{Amplify--and
Forward (AF), Joint Decode--and--Forward (JDF)} and
\emph{Denoise--and--Forward (DNF)}. The two steps are:
\textbf{Step 1:} Nodes $A$ and $C$ transmit, node $B$ receives.
\textbf{Step 2:} Node $B$ transmits, nodes $A$ and $C$ receive.

The transmission rates for $A$ and $C$ in Step 1 are denoted by
$R_A$ and $R_C$, respectively. As we will see, the choice of $R_A$
and $R_C$ is a feature of each transmission scheme AF, JDF or DNF.
Except for the selection of the rate pair $(R_A,R_C)$ rates,  the
Step 1 is identical for all three schemes, where its duration is
fixed to $N$ symbols and the $m-$th received symbol at node $B$ is
given by~(\ref{eq:MAchannel_at_B}).

\subsection{Amplify--and--Forward (AF)}

After Step 1, the node $B$ amplifies the received signal
$\textbf{y}_B$ for a factor $\beta$ and broadcasts $\textbf{x}_B=
\beta \textbf{y}_B$ to $A$ and $C$. As $\textbf{x}_B$ also
consists of $N$ symbols, the total duration of the two steps is
$2N$. The amplification factor $\beta$ is chosen as:
\begin{equation}\label{eq:AF_beta factor}
    \beta=\sqrt{\frac{1}{|h_1|^2+|h_2|^2+N_0}}
\end{equation}
to make the the average per--symbol transmitted energy at $B$
equal to 1 ($N_0$ is the noise variance). The $m-$th symbol
received by $A$ in Step 2 is: {\setlength\arraycolsep{0pt} \small
\begin{eqnarray}
& & y_A[m] = \beta h_1 y_B[m] + z_A[m] = \nonumber \\
& & \beta h^2_1 x_A[m]+ \beta h_1 h_2 x_C[m]+\beta h_1 z_B[m]+
z_A[m] \nonumber
\end{eqnarray}}
Since $A$ knows $x_A[m],h_1, h_2$ and $\beta$, it can subtract
$\beta h^2_1 x_A[m]$ from $y_A[m]$ and obtain: { \small
\begin{equation}\label{eq:AF-A-received-signal}
    r_A[m]=\beta h_1 h_2 x_C[m]+\beta h_1 z_B[m]+ z_A[m]
\end{equation}}
which is a Gaussian channel for receiving $x_C[m]$ with SNR:
\begin{equation}\label{eq:AF-SNR-A}
    \gamma^{(AF)}_{C \rightarrow A}=\frac{\beta^2 |h_1|^2|h_2|^2}{(\beta^2 |h_1|^2+1)
    N_0
    }=\frac{\gamma_1 \gamma_2}{2 \gamma_1 + \gamma_2 +1}
\end{equation}
This notation denotes that $\gamma^{(AF)}_{C \rightarrow A}$ is
the SNR that determines the rate $R_C$ at which $C$ can
communicate to $A$. Similarly, we can find the SNR which
determines the rate $R_A$:
\begin{equation}\label{eq:AF-SNR-C}
    \gamma^{(AF)}_{A \rightarrow C}=\frac{\gamma_1 \gamma_2}{\gamma_1 + 2 \gamma_2 +1}
\end{equation}
Hence, the rate pair $(R_A, R_C)$ used in Step 1 should be:
\begin{equation}\label{eq:AF rates A C}
    R_A=C \left( \gamma^{(AF)}_{A \rightarrow C} \right )
    \qquad R_C=C \left( \gamma^{(AF)}_{C \rightarrow A} \right )
\end{equation}
Finally, the two--way rate achieved by the AF scheme is:
\begin{equation}\label{eq:AF overall data rate}
    R_{AF}=\frac{NR_A+NR_C}{2N}=\frac{R_A+R_C}{2}
\end{equation}

\subsection{Joint Decode--and--Forward (JDF)}

Here the at rates $R_A$ and $R_C$ are chosen such that the node
$B$ is able to decode both packets in Step 1. The rate pairs
$(R_A, R_C)$ with such a property should lie inside the convex
region~\cite{ref:ThomasCover} on Fig.~\ref{fig:MArates}. The
sum--rate is maximized if the rate pair $(R_A,R_C)$ lies on the
segment $\overline{L_AL_C}$:
\begin{equation}\label{eq:JDF MAxJointrate}
    R_A+R_C=C(\gamma_1+\gamma_2)
\end{equation}
while $R_A+R_C<C(\gamma_1+\gamma_2)$ in all other points of the
region of achievable rates. The points $L_A$ and $L_C$ are
determined as: {\small \arraycolsep=2pt
\begin{eqnarray}
R_A(L_A)=C(\gamma_1), R_C(L_A)=C \left( \frac{\gamma_2}{1+\gamma_1} \right) \nonumber \\
R_A(L_C)=C\left ( \frac{\gamma_1}{1+\gamma_2} \right),
R_C(L_C)=C(\gamma_2)
\end{eqnarray}
} For the rate pair at $L_A$, the packet $\textbf{x}_C$ is decoded
first, it is then subtracted from the received signal and then
$\textbf{x}_A$ is decoded. At the point $L_C$, these operations
are reversed. Any other point $L$ on the line $\overline{L_AL_C}$
has rates {\small
\begin{eqnarray}\label{eq:JDFpoint on LaLc}
    R_A(\lambda)= C\left ( \frac{\gamma_1}{1+\gamma_2}
    \right)+\lambda \left( C(\gamma_1)- C\left ( \frac{\gamma_1}{1+\gamma_2}
    \right) \right) \\
    R_C(\lambda)= C(\gamma_2)+\lambda \left( C\left (
    \frac{\gamma_2}{1+\gamma_1} \right) -C(\gamma_2)
     \right)
\end{eqnarray}
} where $0 \leq \lambda \leq 1$ can be the time--sharing
parameter, see~\cite{ref:ThomasCover}.
\begin{figure}
    \centering
    \includegraphics[width=5cm,height=5cm]{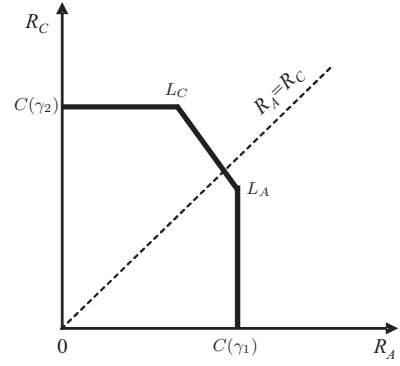}\\
    \caption{The convex hull of the rate pairs $(R_A,
    R_C)$ that are decodable by $B$ in Step 1.
    The dashed line denotes the rate pairs with $R_A=R_C$.
    }\label{fig:MArates}
    \vspace{-12pt}
\end{figure}

\begin{theorem}
\label{theorem:JDF} The maximal two--way rate for the joint
decode--and--forward (JDF) scheme is
\begin{equation}\label{eq:MaxRateJDF}
    R^{*}_{JDF}=
    \left \{
    \begin{array}{ll}
        C(\gamma_1)\frac{2C(\gamma_1+\gamma_2)}{2C(\gamma_1)+C(\gamma_1+\gamma_2)} & \textrm{if $ \gamma_1 \leq \gamma_2 \leq \gamma_1 +\gamma_1^2$}\\
        C(\gamma_1) & \textrm{if $\gamma_2 > \gamma_1 +\gamma_1^2 $}
    \end{array}
    \right.
\end{equation}
\end{theorem}

\begin{proof}
The starting point is the fact that the line segment
$\overline{L_AL_C}$ contains at least one rate pair $(R_A, R_C)$
that maximizes the two--way rate. We omit this proof as it can be
done in a similar way as the part of the proof that follows. We
consider two different cases, one for each region of $\gamma_2$.

\subsubsection{Case $ \gamma_1 \leq \gamma_2 \leq \gamma_1
+\gamma_1^2$} In this region of values for $\gamma_1,\gamma_2$
there is a value $\lambda_0$, such that:
\begin{equation}\label{eq:JDF condition lamda0}
    R_A(\lambda_0)=R_C(\lambda_0)
\end{equation}
i. e. the dashed line on Fig.~\ref{fig:MArates} intersects with
the segment $\overline{L_AL_C}$. The value of $\lambda_0$ is
determined as:
\begin{equation}\label{eq:JDF lamda0}
    \lambda_0=\frac{2C(\gamma_2)-C(\gamma_1+\gamma_2)}{2C(\gamma_1)+2C(\gamma_2)-2C(\gamma_1+\gamma_2)}
\end{equation}
There are two subcases:

\noindent \underline{Subcase} $\lambda \leq \lambda_0$. Here
$R_C(\lambda)>R_A(\lambda)$ and the packet $\textbf{D}_{CA}$ sent
by node $C$ contains more bits than the packet $\textbf{D}_{AC}$.
After decoding both packets, the node $B$ pads the packet
$\textbf{D}_{AC}$ with zeros to obtain $\textbf{D}^{p}_{AC}$ with
$|\textbf{D}^{p}_{AC}|=|\textbf{D}_{CA}|$ and creates:
\begin{equation}\label{eq:JDF XORpacket}
    \textbf{D}_{B}=\textbf{D}^{p}_{AC} \oplus \textbf{D}_{CA}
\end{equation}
Note again that the nodes $A$ and $C$ know a priori how many
padding zeros are used. Since $\gamma_1 \leq \gamma_2$, in Step 2
of the JDF scheme the node $B$ broadcasts $\textbf{D}_{B}$ at a
rate $C(\gamma_1)$. After receiving $\textbf{D}_{B}$, the node $A$
obtains $\textbf{D}_{CA}=\textbf{D}^{p}_{AC} \oplus
\textbf{D}_{B}$ and the node $C$ obtains
$\textbf{D}^{p}_{AC}=\textbf{D}_{CA} \oplus \textbf{D}_{B}$ and
hence obtains $\textbf{D}_{AC}$. The total number of symbols used
in the two steps is $N_{1,
JDF}(\lambda)=N+N\frac{R_C(\lambda)}{C(\gamma_1)}$, such that the
two--way rate is:
\begin{equation}\label{eq:JDF_case1_tw-way rate}
    R_{1,JDF}(\lambda)=\frac{NR_A(\lambda)+NR_C(\lambda)}{N+N\frac{R_C(\lambda)}{C(\gamma_1)}}=C(\gamma_1)\frac{C(\gamma_1+\gamma_2)}{C(\gamma_1)+R_C(\lambda)}
\end{equation}
since~(\ref{eq:JDF MAxJointrate}) holds for each $\lambda$. As
$R_C(\lambda)$ decreases with $\lambda$, the value
$R_{1,DF}(\lambda)$ is maximized for $\lambda=\lambda_0$, where
$\lambda_0$ is given by~(\ref{eq:JDF lamda0}), such that
$R_{1,DF}(\lambda_0)=C(\gamma_1)\frac{2C(\gamma_1+\gamma_2)}{2C(\gamma_1)+C(\gamma_1+\gamma_2)}$.

\noindent \underline{Subcase} $\lambda > \lambda_0$. Here
$R_A(\lambda)>R_C(\lambda)$ and hence
$|\textbf{D}_{AC}|>|\textbf{D}_{CA}|$. The proof uses similar line
of argument as in case 1 of the proof of theorem~\ref{theorem:DF}
and therefore we briefly sketch it. The first part of the packet
$\textbf{D}_{AC}$ is \verb|XOR|--ed with the packet
$\textbf{D}_{CA}$ and the resulting packet is broadcasted at rate
$C(\gamma_1)$. Then, the rest of the packet $\textbf{D}_{AC}$  is
broadcasted at a higher rate $C(\gamma_2)$. The total number of
symbols in the two steps is:
\begin{equation}\label{eq:JDF_subcase2_number of time steps}
    N_{2, JDF}(\lambda)=N+N\frac{R_C(\lambda)}{C(\gamma_1)}+N\frac{R_A(\lambda)-R_C(\lambda)}{C(\gamma_2)}
\end{equation}
This leads to two--way rate of
\begin{equation}\label{eq:JDF_subcase2_tw-way rate}
    R_{2,JDF}(\lambda)=\frac{NC(\gamma_1+\gamma_2)}{N_{2, JDF}(\lambda)}
\end{equation}
It can be shown that $N_{2, JDF}(\lambda)$ is monotonically
decreasing with $\lambda$, while
$R_{2,JDF}(\lambda_0)=R_{1,JDF}(\lambda_0)$, which proves that the
maximal rate is achieved at $\lambda=\lambda_0$.

\subsubsection{Case $ \gamma_2 > \gamma_1
+\gamma_1^2$}. In this case for any $\lambda$, $0 \leq \lambda
\leq 1$ it holds that $R_C(\lambda)>R_A(\lambda)$. Hence, we can
use the transmission method for the subcase $\lambda \leq
\lambda_0$, discussed above. The obtained two--way rate is again
given by~(\ref{eq:JDF_case1_tw-way rate}), which is monotonically
increasing with $\lambda$ and attains the maximum for $\lambda=1$.
Hence, the maximal two--way rate is:
\begin{equation}\label{eq:JDF_case2_tw-way rate}
    R_{1,JDF}(\lambda=1)=C(\gamma_1)\frac{C(\gamma_1+\gamma_2)}{C(\gamma_1)+R_C(\lambda=1)}=C(\gamma_1)
\end{equation}
It can be shown that there are other pairs $R_C, R_A$ that achieve
the maximal two--way rate. Those pairs lie on the segment
$\overline{L_AL_E}$, where $L_E$ is the point where
$R_A=R_C=C(\gamma_1)$.
\end{proof}
Note that $R^{*}_{JDF}< C(\gamma_1)$ when $\gamma_2 < \gamma_1
+\gamma_1^2$.

\subsection{Denoise--and--forward (DNF)}

In the first step of this scheme, the nodes $A$ and $C$ transmit
the packets $\textbf{x}_A$ and $\textbf{x}_C$ at rates $R_A$ and
$R_C$ but we \emph{do not require that the node $B$ is able to
decode} the packets $\textbf{x}_A$ and $\textbf{x}_C$. During the
$N$ symbols of Step1, $B$ receives the $N-$dimensional complex
vector $\textbf{y}_B$, where the $m-$th symbol of $\textbf{y}_B$
is given by~(\ref{eq:MAchannel_at_B}). If the selected rate pair
$(R_A, R_C)$ is not achievable for the multiple access channel
(i.~e. lies outside the convex region on Fig.~\ref{fig:MArates}),
then $B$ cannot find unique pair of codewords $(\textbf{x}_A,
\textbf{x}_C)$, such that the triplet $(\textbf{x}_A,
\textbf{x}_C, \textbf{y}_B)$ is \emph{jointly typical}. The
concept of joint typicality is rather a standard one in
information theory and the reader is referred
to~\cite{ref:ThomasCover} for precise definition. For our
discussion it is sufficient to say that $(\textbf{x}_A,
\textbf{x}_C, \textbf{y}_B)$ is jointly typical when the codeword
$(\textbf{x}_A, \textbf{x}_C)$ is likely to produce $\textbf{y}_B$
at $B$. When the pair $(R_A, R_C)$ is not achievable over the
multiple--access channel, then, upon observing $\textbf{y}_B$, the
node $B$ has a set of codeword pairs ${\cal J}(\textbf{y}_B)$ such
that:
\begin{equation}\label{eq:Jointly typical set}
    {\cal J}(\textbf{y}_B)=\{(\textbf{x}_A, \textbf{x}_C) |
(\textbf{x}_A, \textbf{x}_C, \textbf{y}_B)\textrm{ is jointly
typical}\}
\end{equation}

\begin{lemma}
Let $\textbf{y}_B$ be a typical sequence. Let $(\textbf{x}^{1}_A,
\textbf{x}^{1}_C)$ and $(\textbf{x}^{2}_A, \textbf{x}^{2}_C)$ be
two distinct codeword pairs in ${\cal J}(\textbf{y}_B)$. If $R_A
\leq C(\gamma_1)$ and $R_C \leq C(\gamma_2)$, then $A$ and $C$ can
always select the codebooks such that
\begin{equation}\label{eq:DNF-lemma 1}
    \textbf{x}^{1}_A \neq \textbf{x}^{2}_A\textrm{ and
    }\textbf{x}^{1}_C \neq \textbf{x}^{2}_C
\end{equation}
\end{lemma}
\begin{proof}
If $B$ knows packet of $C$, then $A$ can transmit to $B$ reliably
up to the rate $C(\gamma_1)$. We prove the lemma by contradiction.
Let us assume that the contrary is true: $\textbf{x}^{1}_A \neq
\textbf{x}^{2}_A$ and $\textbf{x}^{1}_C = \textbf{x}^{2}_C$. Now,
assume that, after receiving $\textbf{y}_B$, the node $B$ is told
by a genie--helper which is the codeword $\textbf{x}^{1}_C$. Then,
$B$ would still have ambiguity whether $A$ has sent
$\textbf{x}^{1}_A$ or $\textbf{x}^{2}_A$. But that contradicts the
fact that $A$ can communicate reliably to $B$ at a rate $\leq
C(\gamma_1)$ if $\textbf{x}_C$ is known a priori to $B$.
\end{proof}

From this lemma it follows that, if in Step 2 $B$ manages to send
the exact value $\textbf{y}_B$ (with no additional noise) to $A$
and $C$, then $A$ ($C$) will be able to retrieve the packet sent
by $C$ ($A$) in Step 1. In the DNF scheme the node $B$ maps
$\textbf{y}_B$ to a discrete set of codewords and, in Step 2 it
broadcasts the codeword to which $\textbf{y}_B$ is mapped. Such a
mapping to discrete codewords is referred to as \emph{denoising}.
Let ${\cal Y}_B$ denote the set of typical sequences
$\textbf{y}_B$, each of size $N$. Let ${\cal A}$ be a \emph{set of
denoising codewords} $\{w_B(1), w_B(2), \ldots w_B(|{\cal A}|)
\}$, where $|{\cal A}|$ is the cardinality of the set. The
denoising is defined through the following mapping:
\begin{equation}\label{eq:DNF denoise mapping}
    {\cal D}: \quad {\cal Y}_B \mapsto {\cal A}
\end{equation}
The codewords in ${\cal A}$ are random i.~e. selected in a manner
that achieves the capacity of the associated Gaussian channel.
Upon observing $\textbf{y}_B$ in Step 1, in Step 2 the node $B$
broadcasts the codeword ${\cal D}(\textbf{y}_B)$. The mapping
${\cal D}$ should have the following property:
\begin{property}
Given the codeword ${\cal D}(\textbf{y}_B)$ and with known
codeword $\textbf{x}_A$ ($\textbf{x}_C$), the other codeword
$\textbf{x}_C$ ($\textbf{x}_A$) can be retrieved unambiguously.
\end{property}

Such a property enables $A$ and $C$ to successfully decode each
other's packets after Step 2. The important question is: For given
$(R_A, R_C)$ from Step 1, what should be the minimal size $|{\cal
A}|$, such that Property 1 is satisfied? Assume that $R_C>R_A$,
then there are $2^{NR_C}$ possible codewords that $C$ can send in
Step 1 vs. $2^{NR_A}<2^{NR_C}$ sent by $A$. Clearly, the
cardinality should be at least $|{\cal A}| \leq 2^{NR_C}$, because
otherwise it is impossible for $A$ to reconstruct the codeword
sent by $C$. In this paper we conjecture, without proof, that it
is always possible to design the denoising by using a set of
minimal possible cardinality that can satisfy the Property 1:
\begin{equation}\label{eq:DNF central theorem}
    |{\cal A}|=\max(2^{NR_A},2^{NR_C})
\end{equation}
Such a choice is guaranteed to offer an upper bound on the
two--way rate of DNF and is equal to the achievable rate of DNF if
the conjecture is valid.
\begin{theorem}
The upper bound on the two--way rate for denoise--and--forward
(DNF) is
\begin{equation}\label{eq:MaxRateDNF}
    R^{*}_{DNF}=C(\gamma_1)
\end{equation}
where $\gamma_1$ is the SNR of the weaker link to the relay.
\end{theorem}
\begin{proof}
The rate $R_A=C(\gamma_1)$ is maximal possible, while the rate
$R_C=C(\gamma)$, where $\gamma_1 \leq \gamma \leq \gamma_2$. After
the Step 1, the node $B$ maps the received sequence $\textbf{y}_B$
according to the denoising to ${\cal D}(\textbf{y}_B)$. As there
are $|{\cal A}|=2^{NR_C}$ denoising codewords, each one is
represented by $NR_C$ bits. Since both $A$ and $C$ need to receive
it, the codeword ${\cal D}(\textbf{y}_B)$ needs to be sent at a
rate $C(\gamma_1)$. The total duration of the two steps is
$N_{DNF}=N+N\frac{C(\gamma)}{C(\gamma_1)}$ which makes the
two--way rate:
\begin{equation}\label{eq:R_DNF}
    R^{*}_{DNF}=\frac{NC(\gamma_1)+NC(\gamma)}{N+N\frac{C(\gamma)}{C(\gamma_1)}}=C(\gamma_1)
\end{equation}
\end{proof}
This result implies that the node $C$ does not need to ``fully
load'' the channel by setting $R_C=C(\gamma_2)$ and any value of
$R_C \geq C(\gamma_1)$ will result in the maximal two--way rate.
Hence, the higher transmission rate $R_C$ does not improve the
two--way rate, as it accumulates more data at $B$ which needs to
be broadcasted at a low rate in Step 2. Finally, while the JDF
scheme achieves a two--way rate of $C(\gamma_1)$ only when
$\gamma_2 \geq \gamma_1 +\gamma_1^2$, the DNF scheme achieves it
even for $\gamma_2=\gamma_1$.

\section{Numerical Illustration}

\begin{figure}
    \centering
    \includegraphics[width=8.3 cm]{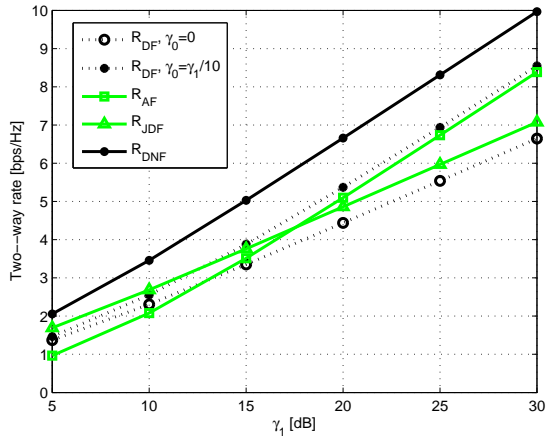}\\
    \caption{Maximal two--way rate for the different schemes with
    $\gamma_2=\gamma_1$.
    }\label{fig:g1_eq_g2}
    \vspace{- 12pt}
\end{figure}
\begin{figure}
    \centering
    \includegraphics[width=8.3 cm]{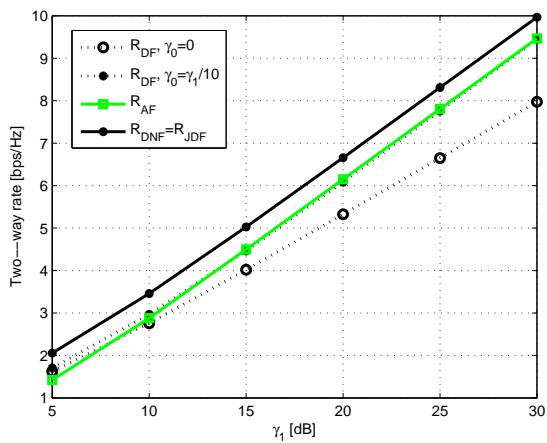}\\
    \caption{Maximal two--way rate for the different schemes with $\gamma_2=\gamma_1+\gamma_1^2$
    }\label{fig:g2_lt_g1}
    \vspace{- 12pt}
\end{figure}

Fig.~\ref{fig:g1_eq_g2} and Fig.~\ref{fig:g2_lt_g1} depict the
two--way rate vs. the SNR $\gamma_1$. In both figures, the $DF$
scheme is evaluated for two different values of the SNR on the
direct link, $\gamma_0=0$ and $\gamma_0=\frac{\gamma_1}{10}$.
Fig.~\ref{fig:g1_eq_g2} shows the results when the SNR of the link
$B-C$ is $\gamma_2=\gamma_1$. As expected, the upper bound
$R_{DNF}$ is always highest for all $\gamma_1$. While $R_{AF}$ is
lower than $R_{JDF}$ for low SNRs, at high SNR the noise
amplification loses significance and thus AF achieves higher
two--way rate than JDF. Also, note that the improvement of the
direct link $\gamma_0$, brings significant increase of the
two--way rate in the DF scheme. Fig.~\ref{fig:g2_lt_g1} shows the
results when $\gamma_2=\gamma_1+\gamma_1^2$, the lowest value for
$\gamma_2$ at which the rate of JDF becomes equal to teh upper
bound for DNF. Clearly, the curve for DNF remains the same as in
Fig.~\ref{fig:g1_eq_g2}, while the increased $\gamma_2$ is
reflected in improved two--way rates for AF and DF. The
improvement is larger for AF, which now slightly outperforms DF
with $\gamma_0=\frac{\gamma_1}{10}$ at higher SNRs.

\section{Conclusion}

We have investigated several methods that implement physical
network coding for two--way relay channel. We have grouped the
physical network coding schemes into two generic groups of 3--step
and 2--step schemes, respectively. The 3--step scheme is
Decode--and--Forward (DF), while we consider are three 2--step
schemes \emph{Amplify--and Forward (AF), Joint
Decode--and--Forward (JDF)} and \emph{Denoise--and--Forward
(DNF)}. We have derived the achievable rates for DF, AF, and JDF,
as well as an upper bound on the achievable rate of DNF. The
numerical results confirm that no scheme can achieve higher
two--way rate than the upper bound of DNF. Nevertheless, there are
certain SNR configurations of the source--relay links under which
the maximal two--way rate of JDF is identical with the uppper
bound of DNF. As a future work, we are first going to provide a
proof that the upper bound for DNF is achievable. Another
important aspect is investigation of the impact that the efficient
broadcasting schemes~\cite{ref:ThomasCover} can have on the DF and
JDF scheme. It is interesting to investigate how to design a
3--step scheme when the direct link is better than one of the
source--relay links. Although some practical DNF methods have been
outlined in~\cite{ref:BAT-ICC}, it is important to investigate how
to perform DNF when different modulation/coding methods are
applied. Finally, a longer--term goal is to investigate how the
physical network coding can be generalized to the scenarios with
multiple communicating nodes and multiple relays.

\bibliographystyle{IEEEtran}
\bibliography{IEEEabrv,ICC07_refs}
\end{document}